\documentclass[letterpaper,12pt]{article}
\usepackage{amssymb, amsmath}
\usepackage{amsthm}
\usepackage{algorithm}
\usepackage{algorithmic}
\usepackage{enumerate}
\usepackage[numbers,square,comma,sort&compress]{natbib}
\usepackage[pdftex,colorlinks]{hyperref}
\usepackage{a4}

\newtheorem{theorem}{Theorem}
\newtheorem{fact}[theorem]{Fact}
\newtheorem{lemma}[theorem]{Lemma}

\newtheorem*{chebyshev}{Chebyshev's inequality}
\newcommand{\comment}[1]{}

\begin{document}
\title{A Las Vegas approximation algorithm for metric $1$-median selection}

\author{
Ching-Lueh Chang \footnote{Department of Computer Science and
Engineering,
Yuan Ze University, Taoyuan, Taiwan. Email:
clchang@saturn.yzu.edu.tw}
\footnote{Supported
in part by the Ministry of Science and Technology of Taiwan under
grant 105-2221-E-155-047-.}
}


\maketitle

\begin{abstract}
Given
an $n$-point
metric
space,
consider the problem of finding
a point with the minimum
sum of distances
to all points.
We show that
this problem
has a randomized algorithm that
{\em always} outputs a $(2+\epsilon)$-approximate solution
in an expected $O(n/\epsilon^2)$ time
for each constant $\epsilon>0$.
Inheriting Indyk's~\cite{Ind00} algorithm,
our
algorithm
outputs
a $(1+\epsilon)$-approximate
$1$-median
in $O(n/\epsilon^2)$ time
with probability $\Omega(1)$.
\comment{ 
As a by-product,
we present a Monte-Carlo $O(n/\epsilon^2)$-time algorithm for estimating
the sum of all distances in an $n$-point metric space
to within
a ratio in $[\,1-\epsilon,1+\epsilon\,]$, for each constant $\epsilon>0$.
}
\end{abstract}

\section{Introduction}

A metric space is a nonempty set $M$ endowed with a
metric,
i.e.,
a function
$d\colon M\times M\to[\,0,\infty\,)$
such that
\begin{itemize}
\item $d(x,y)=0$ if and only if $x=y$ (identity of indiscernibles),
\item $d(x,y)=d(y,x)$ (symmetry), and
\item $d(x,y)+d(y,z)\ge d(x,z)$ (triangle inequality)
\end{itemize}
for all
$x$, $y$, $z\in M$~\cite{Rud76}.

For all $n\in\mathbb{Z}^+$, define $[n]\equiv\{1,2,\ldots,n\}$.
Given
$n\in\mathbb{Z}^+$
and oracle access to a metric $d\colon [n]\times[n]\to[\,0,\infty\,)$,
{\sc metric $1$-median}
asks for
$\mathop{\mathrm{argmin}}_{y\in[n]}\,\sum_{x\in[n]}\,d(y,x)$,
breaking ties arbitrarily.
It generalizes the classical median selection on the real line
and has a
brute-force
$\Theta(n^2)$-time algorithm.
More generally, {\sc metric $k$-median} asks for
$c_1$, $c_2$, $\ldots$,
$c_k\in[n]$
minimizing
$\sum_{x\in[n]}\,\min_{i=1}^k\,d(x,c_i)$.
Because
$d(\cdot,\cdot)$ defines
$\binom{n}{2}=\Theta(n^2)$ nonzero distances,
only $o(n^2)$-time algorithms are said to
run in
sublinear time~\cite{Ind99}.
For all $\alpha\ge1$,
an
$\alpha$-approximate
$1$-median
is a point $p\in[n]$
satisfying
$$\sum_{x\in[n]}\,d\left(p,x\right)\le\alpha\cdot
\min_{y\in[n]}\,\sum_{x\in[n]}\,d\left(y,x\right).$$

For all $\epsilon>0$,
{\sc metric $1$-median}
has
a Monte Carlo $(1+\epsilon)$-approximation
$O(n/\epsilon^2)$-time
algorithm~\cite{Ind99, Ind00}.
Guha et al.~\cite{GMMMO03}
show that {\sc metric $k$-median}
has
a Monte Carlo, $O(\exp(O(1/\epsilon)))$-approximation,
$O(nk\log n)$-time, $O(n^{\epsilon})$-space and one-pass algorithm for all small $k$
as well as a deterministic, $O(\exp(O(1/\epsilon)))$-approximation,
$O(n^{1+\epsilon})$-time, $O(n^{\epsilon})$-space and one-pass algorithm.
Given
$n$ points in $\mathbb{R}^D$ with $D\ge 1$,
the Monte Carlo algorithms of Kumar et al.~\cite{KSS10}
find
a $(1+\epsilon)$-approximate $1$-median
in
$O(D\cdot\exp(1/\epsilon^{O(1)}))$ time
and a $(1+\epsilon)$-approximate solution to {\sc metric $k$-median}
in $O(Dn\cdot\exp((k/\epsilon)^{O(1)}))$ time.
All randomized $O(1)$-approximation algorithms for {\sc metric $k$-median}
take $\Omega(nk)$ time~\cite{MP04, GMMMO03}.
Chang~\cite{Cha15}
shows that {\sc metric $1$-median} has
a deterministic, $(2h)$-approximation, $O(hn^{1+1/h})$-time
and nonadaptive algorithm for
all
constants
$h\in\mathbb{Z}^+\setminus\{1\}$, generalizing the results of Chang~\cite{Cha13} and
Wu~\cite{Wu14}.
On the other hand,
he
disproves
the
existence
of
deterministic $(2h-\epsilon)$-approximation $O(n^{1+1/(h-1)}/h)$-time
algorithms
for all constants $h\in\mathbb{Z}^+\setminus\{1\}$ and $\epsilon>0$~\cite{Cha16COCOON, Cha17}.

In social network analysis, the closeness centrality of a point $v$
is
the reciprocal of the
average distance
from $v$ to all
points~\cite{WF94}.
So {\sc metric $1$-median}
asks for
a point with the maximum closeness
centrality.
Given oracle access to a graph metric,
the Monte-Carlo algorithms of
Goldreich and Ron~\cite{GR08} and Eppstein and Wang~\cite{EW04}
estimate the closeness centrality of a given point and those of all points, respectively.

All known
sublinear-time
algorithms
for {\sc metric $1$-median}
are
either deterministic or
Monte Carlo,
the latter having
a positive probability of failure.
For example, Indyk's Monte Carlo $(1+\epsilon)$-approximation algorithm
outputs
with a positive probability
a solution
without approximation guarantees.
In contrast,
we show
that {\sc metric $1$-median} has
a randomized
algorithm
that {\em always} outputs a
$(2+\epsilon)$-approximate solution
in
expected $O(n/\epsilon^2)$
time
for all constants $\epsilon>0$.
So,
excluding
the
known
deterministic algorithms (which are Las Vegas only in the degenerate
sense),
this paper gives
the {\em first} Las Vegas approximation algorithm for {\sc metric $1$-median}
with an expected
sublinear
running
time.
Note that
deterministic
sublinear-time
algorithms for {\sc metric $1$-median}
can be $4$-approximate but not $(4-\epsilon)$-approximate for any constant
$\epsilon>0$~\cite{Cha13, Cha17}.
So our
approximation ratio
of $2+\epsilon$
beats
that of
any
deterministic
sublinear-time
algorithm.
Inheriting
Indyk's algorithm,
our algorithm
outputs a $(1+\epsilon)$-approximate $1$-median in
$O(n/\epsilon^2)$ time with probability $\Omega(1)$ for all constants $\epsilon>0$.
\comment{ 
In case our algorithm
fails to
output a
$(1+\epsilon)$-approximate $1$-median, it
nonetheless
outputs a $(2+\epsilon)$-approximate
$1$-median.
}

\comment{ 
As a by-product of our derivations,
we present a Monte Carlo $O(n/\epsilon^2)$-time algorithm for estimating $\sum_{u,v\in [n]}\,d(u,v)$
to within an additive error of $\epsilon\cdot\sum_{u,v\in [n]}\,d(u,v)$
with an $\Omega(1)$ probability of success, for each constant $\epsilon>0$.
Previously, the best algorithm for such estimation
needs
$O(n/\epsilon^{7/2})$ time~\cite[Sec.~8]{Ind99}.
}

\comment{ 
So
we have
the first Las Vegas
approximation
algorithm
for {\sc metric $1$-median}
with an expected sublinear running time.
Because
the best approximation ratio
achievable by deterministic sublinear-time algorithms for {\sc metric $1$-median}
is $4$~\cite{Cha13, Cha15},
Las Vegas
approximation
algorithms
have a better approximation ratio.
}

Below is
our
high-level
and inaccurate
sketch of proof,
where $\epsilon$,
$\delta>0$
are small constants:
\begin{enumerate}[(i)]
\item\label{thefirststep_runIndyk} Run Indyk's algorithm to find a probably $(1+\epsilon/10^{10})$-approximate $1$-median, $z$.
Then let $r=\sum_{x\in[n]}\,d(z,x)/n$ be the average distance from $z$
to all points.
\item\label{farawaypointsaredirectlyruledout}
For all $R>0$, denote
by $B(z,R)$ the open ball with center $z$ and radius $R$.
Use the triangle inequality (with details omitted here) to show
$z$ to be a
solution
no worse
than the points
in $[n]\setminus B(z,8r)$,
i.e.,
\begin{eqnarray}
\sum_{x\in[n]}\,d\left(z,x\right) \le
\inf_{y\in[n]\setminus B(z,8r)}\,
\sum_{x\in[n]}\,d\left(y,x\right).
\label{farawaypointsarebadanyway}
\end{eqnarray}
\item\label{lowerbounditem}
Take
a uniformly random bijection
$\pi\colon [\,|B(z,\delta nr)|\,]\to B(z,\delta nr)$.
Then observe that
{\small 
\begin{eqnarray}
\min_{y\in B(z,8r)}\,
\sum_{x\in B(z,\delta nr)}\,d\left(y,x\right)
&\ge&
\min_{y\in B(z,8r)}\,
\sum_{i=1}^{\lfloor|B(z,\delta nr)|/2\rfloor}\,
\left(d\left(y,\pi\left(2i-1\right)\right)+d\left(y,\pi\left(2i\right)\right)\right)
\,\,\,\,\,\,\,\,\,\label{anyclosenesscentralityhassizeatleastthatofamatchingfirst}\\
&\ge&
\sum_{i=1}^{\lfloor|B(z,\delta nr)|/2\rfloor}\, d\left(\pi\left(2i-1\right),\pi\left(2i\right)\right),
\label{anyclosenesscentralityhassizeatleastthatofamatching}
\end{eqnarray}
}
where the first (resp., second) inequality follows from the injectivity of $\pi$ (resp.,
the triangle inequality).
\item\label{thehardeststep} Assume $B(z,\delta nr)=[n]$ for simplicity.
So by
inequalities~(\ref{farawaypointsarebadanyway})--(\ref{anyclosenesscentralityhassizeatleastthatofamatching}),
if the following inequality holds, then it serves as a
witness that $z$ is $(2+\epsilon)$-approximate:
\begin{eqnarray}
\sum_{x\in B(z,\delta nr)}\, d\left(z,x\right)\le
\left(2+\epsilon\right)\cdot \sum_{i=1}^{\lfloor|B(z,\delta nr)|/2\rfloor}\,
d\left(\pi\left(2i-1\right),\pi\left(2i\right)\right).
\label{intuitivegoal}
\end{eqnarray}
\footnote{Assuming $B(z,\delta nr)=[n]$,
inequalities~(\ref{anyclosenesscentralityhassizeatleastthatofamatchingfirst})--(\ref{intuitivegoal})
imply $\sum_{x\in[n]}\,d(z,x)\le (2+\epsilon)\cdot\sum_{x\in[n]}\,d(y,x)$
for all $y\in B(z,8r)$.
Furthermore,
$\sum_{x\in[n]}\,d(z,x)\le \sum_{x\in[n]}\,d(y,x)$
for all $y\in [n]\setminus B(z,8r)$ by inequality~(\ref{farawaypointsarebadanyway}).}
To guarantee outputting a $(2+\epsilon)$-approximate
$1$-median,
output $z$
only when inequality~(\ref{intuitivegoal}) holds.
Restart from item~(\ref{thefirststep_runIndyk}) whenever
inequality~(\ref{intuitivegoal}) is false.
\end{enumerate}

More
details
of
item~(\ref{thehardeststep}) follow:
For a $1$-median $z'$ of $B(z,\delta nr)$,
it will be easy to
show
\begin{eqnarray}
\sum_{x\in B(z,\delta nr)}\, d\left(z',x\right)\le
\left(2+o(1)\right)\cdot \mathop{\mathrm E}\left[
\sum_{i=1}^{\lfloor|B(z,\delta nr)|/2\rfloor}\, d\left(\pi\left(2i-1\right),\pi\left(2i\right)\right)
\right].
\label{intuitiveeq1}
\end{eqnarray}
\footnote{Though
not
directly stated in later sections,
this
is a consequence of
Lemmas~\ref{squareofthemeanlemma}~and~\ref{1medianinthesmallerballhasatmosttheaveragedistance}
in Sec.~\ref{expectedtimesection}.}
When $z$ in item~(\ref{thefirststep_runIndyk}) is indeed $(1+\epsilon/10^{10})$-approximate,
\begin{eqnarray}
\sum_{x\in [n]}\,
d\left(z,x\right)\le
\left(1+\frac{\epsilon}{10^{10}}\right)\cdot
\sum_{x\in [n]}\,
d\left(z',x\right).
\label{intuitiveeq2}
\end{eqnarray}
Assuming $B(z,\delta nr)=[n]$,
inequalities~(\ref{intuitiveeq1})--(\ref{intuitiveeq2})
make
inequality~(\ref{intuitivegoal})
hold with high probability
as long as
$\sum_{i=1}^{\lfloor|B(z,\delta nr)|/2\rfloor}\, d(\pi(2i-1),\pi(2i))$
is highly
concentrated
around its expectation.
The need
for such concentration
is why
we
restrict
the radius of
the codomain
of
$\pi$
to be $\delta nr$
in item~(\ref{lowerbounditem})---Large
distances
ruin concentration bounds.
To
accommodate for
the
points in
$[n]\setminus B(z,\delta nr)$,
our
witness for the approximation ratio of $z$
actually differs slightly from
inequality~(\ref{intuitivegoal}),
unlike in item~(\ref{thehardeststep}).\footnote{Our witness for the approximation
ratio of $z$ is as in line~6 of {\sf Las Vegas median}\ in Fig.~\ref{mainalgorithm}.}

\section{Definitions and preliminaries}

For a metric space $([n],d)$, $x\in[n]$ and $R>0$,
define
$$
B\left(x,R\right)\equiv\left\{y\in[n]\mid d\left(x,y\right)<R\right\}
$$
to be the open ball with center $x$ and radius $R$.
For brevity,
$$
B^2\left(x,R\right)\equiv B\left(x,R\right)\times B\left(x,R\right).
$$
The
pairs in $B^2(x,R)$ are ordered.

An algorithm $A$ with oracle access to
$d\colon [n]\times[n]\to[\,0,\infty\,)$
is denoted by $A^d$
and
may query $d$ on any $(x,y)\in[n]\times[n]$ for $d(x,y)$.
In this paper, all
Landau symbols (such as $O(\cdot)$, $o(\cdot)$, $\Theta(\cdot)$ and $\Omega(\cdot)$)
are w.r.t.\ $n$.
The following
result
is due to Indyk.

\begin{fact}[\cite{Ind99, Ind00}]\label{Indykfact}
For all $\epsilon>0$,
{\sc metric $1$-median} has a Monte Carlo $(1+\epsilon)$-approximation
$O(n/\epsilon^2)$-time algorithm
with
a failure probability of at most
$1/e$.
\end{fact}

Henceforth,
denote Indyk's algorithm in Fact~\ref{Indykfact} by {\sf Indyk median}.
It is given
$n\in\mathbb{Z}^+$,
$\epsilon>0$
and oracle access to
a metric
$d\colon [n]\times[n]\to[\,0,\infty\,)$.
By convention, denote the expected value and the variance of a random variable $X$
by $\mathop{\mathrm{E}}[\,X\,]$ and $\mathop{\mathrm{var}}(X)$, respectively.

\comment{ 
\begin{fact}[Implicit in~{\cite[Theorem~22]{Cha12}}]\label{uniformlyrandompointisgood}
A uniformly random point of $[n]$ is a $4$-approximate $1$-median
with probability at least $1/2$.
\end{fact}
}

\begin{chebyshev}
[\cite{MR95}]
Let
$X$
be a random variable
with a finite expected value and a finite
nonzero variance.
Then for all
$k\ge1$,
$$
\Pr\left[\,
\left|\,
X-\mathop{\mathrm E}[X]\,\right|\ge k\sqrt{\mathop{\mathrm{var}}(X)}
\,\right]\le \frac{1}{k^2}.
$$
\end{chebyshev}

\section{Algorithm and approximation ratio}\label{approximationratiosection}

\begin{figure}
\begin{algorithmic}[1]
\STATE Find $\delta>0$ such that $2+\epsilon=2/(1-100\sqrt{\delta})$;
\WHILE{{\sf true}}
  \STATE $z\leftarrow\text{\sf Indyk median}^d(n,\epsilon/10^{10})$;
  \STATE $r\leftarrow \sum_{x\in [n]}\,d(z,x)/n$;
  \STATE Pick a
uniformly random
bijection
$\pi\colon [\,|B(z,\delta nr)|\,]\to B(z,\delta nr)$;
  \IF{$\sum_{i=1}^{\lfloor|B(z,\delta nr)|/2\rfloor}\,
d(\pi(2i-1),\pi(2i))+\sum_{x\in [n]\setminus B(z,\delta nr)}\,(d(z,x)-8r)
\ge(1-100\sqrt{\delta})nr/2$}
    \RETURN $z$;
  \ENDIF
\ENDWHILE
\end{algorithmic}
\caption{Algorithm {\sf Las Vegas median} with oracle access to a metric
$d\colon [n]\times[n]\to[\,0,\infty\,)$
and with inputs $n\in\mathbb{Z}^+$ and a small constant
$\epsilon>0$}
\label{mainalgorithm}
\end{figure}

Throughout this
paper,
take
any small constant
$\epsilon>0$,
e.g., $\epsilon=10^{-100}$.
By line~1 of {\sf Las Vegas median}
in Fig.~\ref{mainalgorithm}, $\delta>0$
is likewise a small constant.
The following lemma
implies
that
$z$ in line~3 of {\sf Las Vegas median}
is a solution (to {\sc metric $1$-median}) no worse than
those
in
$[n]\setminus B(z,8r)$, where $r$ is as in line~4.

\begin{lemma}\label{pointsoutofballarebad}
In each iteration of the {\bf while} loop of {\sf Las Vegas median},
$$
\inf_{y\in [n]\setminus B(z,8r)}\,\sum_{x\in [n]}\,d(y,x)
\geq 7\cdot \sum_{x\in [n]}\,d(z,x).
$$
\end{lemma}
\begin{proof}
For each $y\in [n]\setminus B(z,8r)$,
\begin{eqnarray*}
\sum_{x\in [n]}\,d(y,x)&\ge& \sum_{x\in [n]}\,\left(d(y,z)-d(z,x)\right)\\
&\ge&\sum_{x\in [n]}\,\left(8r-d(z,x)\right)\\
&=&8nr-\sum_{x\in [n]}\,d(z,x)\\
&=&7\sum_{x\in [n]}\,d(z,x),
\end{eqnarray*}
where the first inequality follows from the triangle inequality,
the second
follows
from $y\notin B(z,8r)$
and the last equality
follows
from line~4 of {\sf Las Vegas median}.
\end{proof}

\begin{lemma}\label{pointsinballarenottoogood}
When line~7 of {\sf Las Vegas median}~is run,
$$
\min_{y\in B(z,8r)}\,\sum_{x\in [n]}\,d(y,x)
\ge \frac{1-100\sqrt{\delta}}{2}\cdot\sum_{x\in [n]}\,d(z,x).
$$
\end{lemma}
\begin{proof}
Pick any $y\in B(z,8r)$.
\comment{ 
Clearly,
$$
\sum_{x\in [n]}\,d(y,x)
=\sum_{x\in B(z,\delta nr)}\,d(y,x)
+\sum_{x\in [n]\setminus B(z,\delta nr)}\,d(y,x).
$$
}
We have
\begin{eqnarray}
\sum_{x\in B(z,\delta nr)}\,d(y,x)
&\ge&\sum_{i=1}^{\lfloor|B(z,\delta nr)|/2\rfloor}\,
\left(d\left(y,\pi\left(2i-1\right)\right)
+d\left(y,\pi\left(2i\right)\right)\right)\label{distancetoballpoints}\\
&\ge&\sum_{i=1}^{\lfloor|B(z,\delta nr)|/2\rfloor}\,
d\left(\pi\left(2i-1\right),\pi\left(2i\right)\right),\nonumber
\end{eqnarray}
where the first
and
the
second
inequalities
follow from
the injectivity of $\pi$ in
line~5
of {\sf Las Vegas median}
and the triangle inequality, respectively.\footnote{Note that
$\pi(1)$, $\pi(2)$, $\ldots$, $\pi(2\,\lfloor|B(z,\delta nr)|/2\rfloor)$
are distinct elements of $B(z,\delta nr)$.}
Furthermore,
\begin{eqnarray}
\sum_{x\in [n]\setminus B(z,\delta nr)}\,d(y,x)
&\ge&\sum_{x\in [n]\setminus B(z,\delta nr)}\,\left(d(z,x)-d(y,z)\right)
\nonumber\\
&\ge&\sum_{x\in [n]\setminus B(z,\delta nr)}\,\left(d(z,x)-8r\right),
\label{distancetooutofballpoints}
\end{eqnarray}
where the first and the second inequalities follow from the triangle
inequality and $y\in B(z,8r)$, respectively.
Summing up
inequalities~(\ref{distancetoballpoints})--(\ref{distancetooutofballpoints}),
$$
\sum_{x\in [n]}\,d(y,x)
\ge
\sum_{i=1}^{\lfloor|B(z,\delta nr)|/2\rfloor}\,
d\left(\pi\left(2i-1\right),\pi\left(2i\right)\right)
+\sum_{x\in [n]\setminus B(z,\delta nr)}\,\left(d(z,x)-8r\right).
$$
This and
lines~6--7
of {\sf Las Vegas median}
imply
$$
\sum_{x\in [n]}\,d(y,x)
\ge\frac{1-100\sqrt{\delta}}{2}\cdot nr
$$
when line~7 is run.
Finally, $nr=\sum_{x\in[n]}\,d(z,x)$ by line~4.
\end{proof}

Lemmas~\ref{pointsoutofballarebad}--\ref{pointsinballarenottoogood}
and line~1 of {\sf Las Vegas median}
yield the following.

\begin{lemma}\label{wehaveagoodsolutionattheend}
When line~7 of {\sf Las Vegas median} is run,
$$
\left(2+\epsilon\right)
\cdot
\min_{y\in [n]}\,\sum_{x\in [n]}\,d(y,x)
\ge
\sum_{x\in [n]}\,d(z,x),
$$
i.e., $z$ is a
$(2+\epsilon)$-approximate
$1$-median.
\end{lemma}

By Lemma~\ref{wehaveagoodsolutionattheend},
{\sf Las Vegas median}
outputs a $(2+\epsilon)$-approximate
$1$-median
at termination.

\section{Probability of termination in any iteration}\label{expectedtimesection}

This section analyzes
the probability of running line~7
in
any
particular
iteration of the {\bf while} loop
of {\sf Las Vegas median}.
The following lemma
uses an easy averaging argument.

\begin{lemma}\label{thesmallradiusballislarge}
$$\left|\,[n]\setminus B\left(z,\delta nr\right)\,\right|\le \frac{1}{\delta}$$
and, therefore,
$$\left|B\left(z,\delta nr\right)\right|
\ge n-\frac{1}{\delta}
=\left(1-o(1)
\right)n.$$
\end{lemma}
\begin{proof}
Clearly,
$$
\sum_{x\in [n]}\,d\left(z,x\right)
\ge \sum_{x\in [n]\setminus B(z,\delta nr)}\,d\left(z,x\right)
\ge \sum_{x\in [n]\setminus B(z,\delta nr)}\,\delta nr
=\left|\,[n]\setminus B\left(z,\delta nr\right)\,\right|\cdot \delta nr.
$$
Then use
line~4 of {\sf Las Vegas median}.
\end{proof}

Henceforth,
assume $n\ge 1/\delta+4$ without loss of generality; otherwise, find
a $1$-median by brute force.
So $|B(z,\delta nr)|\ge 4$ by Lemma~\ref{thesmallradiusballislarge}.
Define
\begin{eqnarray}
r'
\equiv
\frac{1}{|B(z,\delta nr)|^2}
\cdot
\sum_{u, v\in B(z,\delta nr)}\,
d\left(u,v\right)
\label{smallerballaveragedefinition}
\end{eqnarray}
to
be the average distance
in $B(z,\delta nr)$.
\comment{ 
As $z\in B(z,\delta nr)$, the
denominator in the right-hand side of
equation~(\ref{smallerballaveragedefinition})
is nonzero.
}

\begin{lemma}\label{inneraverageandoverallaverage}
$r'\leq 2r$.
\end{lemma}
\begin{proof}
By
equation~(\ref{smallerballaveragedefinition}) and
the triangle inequality,
\begin{eqnarray}
r'
&\le&
\frac{1}{|B(z,\delta nr)|^2}
\cdot
\sum_{u, v\in B(z,\delta nr)}\,
\left(d\left(z,u\right)+d\left(z,v\right)\right)
\label{frominnerdistancetowholedistance}\\
&=&
\frac{1}{|B(z,\delta nr)|^2}
\cdot
\left|B(z,\delta nr)\right|\cdot\left(
\sum_{u\in B(z,\delta nr)}\,
d\left(z,u\right)
+\sum_{v\in B(z,\delta nr)}\, d\left(z,v\right)
\right)\nonumber\\
&=&
\frac{2}{|B(z,\delta nr)|}
\cdot
\sum_{u\in B(z,\delta nr)}\,
d\left(z,u\right).\nonumber
\end{eqnarray}
Obviously,
the average distance from $z$ to the points in $B(z,\delta nr)$
is at most
that from $z$ to all points,
i.e.,
\begin{eqnarray}
\frac{1}{|B(z,\delta nr)|}
\cdot
\sum_{u\in B(z,\delta nr)}\, d\left(z,u\right)
\le
\frac{1}{n}\cdot
\sum_{u\in [n]}\, d\left(z,u\right).
\label{frominnerdistancetowholedistance2}
\end{eqnarray}
Inequalities~(\ref{frominnerdistancetowholedistance})--(\ref{frominnerdistancetowholedistance2})
and
line~4 of {\sf Las Vegas median}
complete the proof.
\end{proof}

To analyze the probability
that
the condition in line~6
of {\sf Las Vegas median}
holds,
we
shall
derive a concentration bound
for
$$
\sum_{i=1}^{\lfloor|B(z,\delta nr)|/2\rfloor}\,
d\left(\pi\left(2i-1\right),\pi\left(2i\right)\right),
$$
whose
expected value and
variance
are
examined
in the
next
four lemmas.

\begin{lemma}\label{squareofthemeanlemma}
With expectations taken over $\pi$,
\begin{eqnarray}
\mathop{\mathrm E}\left[\,
\sum_{i=1}^{\lfloor|B(z,\delta nr)|/2\rfloor}\,
d\left(\pi\left(2i-1\right),\pi\left(2i\right)\right)
\,\right]
=
\frac{1}{2}
\cdot
\left(
1
\pm o(1)
\right)
nr'.
\comment{ 
\frac{1}{|B(z,\delta nr)|\cdot(|B(z,\delta nr)|-1)}
\cdot
\sum_{\text{\rm distinct $u$, $v\in B(z,\delta nr)$}}\,d\left(u,v\right).
}
\end{eqnarray}
\end{lemma}
\begin{proof}
For each $i\in[\,\lfloor|B(z,\delta nr)|/2\rfloor\,]$,
$\{\pi(2i-1),\pi(2i)\}$ is a uniformly random
size-$2$ subset of $B(z,\delta nr)$
by line~5 of {\sf Las Vegas median}.
Therefore,
\begin{eqnarray}
\mathop{\mathrm E}\left[\,
d\left(\pi\left(2i-1\right),\pi\left(2i\right)\right)
\,\right]
&=&
\frac{1}{|B(z,\delta nr)|\cdot(|B(z,\delta nr)|-1)}
\cdot
\sum_{\text{\rm distinct $u$, $v\in B(z,\delta nr)$}}\,d\left(u,v\right)
\,\,\,\,\,\,\,\,\,
\label{startofequivalentstatementontheaverageinnerdistances}\\
&=&
\frac{1}{|B(z,\delta nr)|\cdot(|B(z,\delta nr)|-1)}
\cdot
\sum_{u, v\in B(z,\delta nr)}\,d\left(u,v\right)\nonumber\\
&=&
\left(1+o(1)
\right)r',
\label{equivalentstatementontheaverageinnerdistances}
\end{eqnarray}
where the
second (resp., last) equality follows from
the identity of indiscernibles
(resp.,
equation~(\ref{smallerballaveragedefinition}) and
Lemma~\ref{thesmallradiusballislarge}).
Finally, use
equations~(\ref{startofequivalentstatementontheaverageinnerdistances})--(\ref{equivalentstatementontheaverageinnerdistances}),
the linearity of expectation
and Lemma~\ref{thesmallradiusballislarge}.
\end{proof}

Clearly,
\begin{eqnarray}
&&\mathop{\mathrm E}\left[\,
\left(\sum_{i=1}^{\lfloor|B(z,\delta nr)|/2\rfloor}\,
d\left(\pi\left(2i-1\right),\pi\left(2i\right)\right)\right)^2
\,\right]
\label{startofthemeanofsquare}\\
&=&
\mathop{\mathrm E}\left[\,
\sum_{i=1}^{\lfloor|B(z,\delta nr)|/2\rfloor}\,
d\left(\pi\left(2i-1\right),\pi\left(2i\right)\right)
\cdot \sum_{j=1}^{\lfloor|B(z,\delta nr)|/2\rfloor}\,
d\left(\pi\left(2j-1\right),\pi\left(2j\right)\right)
\,\right]\nonumber\\
&=&
\sum_{\text{distinct $i,j=1$}}^{\lfloor|B(z,\delta nr)|/2\rfloor}\,
\mathop{\mathrm E}\left[\,
d\left(\pi\left(2i-1\right),\pi\left(2i\right)\right)
\cdot d\left(\pi\left(2j-1\right),\pi\left(2j\right)\right)
\,\right]\nonumber\\
&+&
\sum_{i=1}^{\lfloor|B(z,\delta nr)|/2\rfloor}\,
\mathop{\mathrm E}\left[\,
d^2\left(\pi\left(2i-1\right),\pi\left(2i\right)\right)
\,\right],\label{endofthemeanofsquare}
\end{eqnarray}
where the last equality follows from the linearity of expectation
and the separation of pairs $(i,j)$
according to whether $i=j$.

\begin{lemma}\label{distinctdistancesproduct}
With expectations taken over $\pi$,
\begin{eqnarray*}
\sum_{\text{\rm distinct $i,j=1$}}^{\lfloor|B(z,\delta nr)|/2\rfloor}\,
\mathop{\mathrm E}\left[\,
d\left(\pi\left(2i-1\right),\pi\left(2i\right)\right)
\cdot d\left(\pi\left(2j-1\right),\pi\left(2j\right)\right)
\,\right]
\le\frac{1}{4}\cdot \left(1
+
o(1)
\right)n^2\left(r'\right)^2.
\end{eqnarray*}
\end{lemma}
\begin{proof}
Pick any distinct $i$, $j\in[\,\lfloor|B(z,\delta nr)|/2\rfloor\,]$.
By line~5 of {\sf Las Vegas median},
$$\left\{\pi\left(2i-1\right), \pi\left(2i\right), \pi\left(2j-1\right),
\pi\left(2j\right)\right\}$$
is a uniformly random size-$4$ subset of $B(z,\delta nr)$.
So
\begin{eqnarray*}
&&\mathop{\mathrm E}\left[\,
d\left(\pi(2i-1),\pi(2i)\right)
\cdot d\left(\pi(2j-1),\pi(2j)\right)
\,\right]\\
&=&
\frac{1}{|B(z,\delta nr)|
\cdot(|B(z,\delta nr)|-1)\cdot(|B(z,\delta nr)|-2)
\cdot(|B(z,\delta nr)|-3)}\\
&\cdot&
\sum_{\text{distinct $u$, $v$, $x$, $y\in B(z,\delta nr)$}}\,
d\left(u,v\right)\cdot d\left(x,y\right).
\end{eqnarray*}

Clearly,
\begin{eqnarray*}
\sum_{\text{distinct $u$, $v$, $x$, $y\in B(z,\delta nr)$}}\,
d\left(u,v\right)\cdot d\left(x,y\right)
&\le&
\sum_{u, v, x, y\in B(z,\delta nr)}\,
d\left(u,v\right)\cdot d\left(x,y\right)\\
&=&
\sum_{u, v\in B(z,\delta nr)}\,
d\left(u,v\right)
\cdot
\sum_{x, y\in B(z,\delta nr)}\,
d\left(x,y\right)\\
&=&
\left(\sum_{u, v\in B(z,\delta nr)}\,
d\left(u,v\right)
\right)^2.
\end{eqnarray*}
In summary,
\begin{eqnarray*}
&&
\sum_{\text{\rm distinct $i,j=1$}}^{\lfloor|B(z,\delta nr)|/2\rfloor}\,
\mathop{\mathrm E}\left[\,
d\left(\pi\left(2i-1\right),\pi\left(2i\right)\right)
\cdot d\left(\pi\left(2j-1\right),\pi\left(2j\right)\right)
\,\right]\\
&\le&
\left\lfloor\frac{|B(z,\delta nr)|}{2}\right\rfloor
\left(\left\lfloor\frac{|B(z,\delta nr)|}{2}\right\rfloor-1\right)\\
&\cdot&
\frac{1}{|B(z,\delta nr)|
\cdot(|B(z,\delta nr)|-1)\cdot(|B(z,\delta nr)|-2)
\cdot(|B(z,\delta nr)|-3)}\\
&\cdot&
\left(
\sum_{u, v\in B(z,\delta nr)}\,
d\left(u,v\right)
\right)^2.
\end{eqnarray*}
Together with
Lemma~\ref{thesmallradiusballislarge}
and equation~(\ref{smallerballaveragedefinition}),
this completes
the proof.
\end{proof}

\comment{ 
Lemmas~\ref{thesmallradiusballislarge}~and~\ref{distinctdistancesproduct}
and equation~(\ref{smallerballaveragedefinition})
yield the following.

\begin{lemma}
\begin{eqnarray*}
\sum_{\text{\rm distinct $i,j=1$}}^{|B(z,\delta nr)|/2}\,
\mathop{\mathrm E}\left[\,
d\left(\pi(2i-1),\pi(2i)\right)
\cdot d\left(\pi(2j-1),\pi(2j)\right)
\,\right]
=\frac{1}{4}\cdot n^2\left(r'\right)^2\left(1\pm o(1)\right).
\end{eqnarray*}
\end{lemma}
}


\begin{lemma}\label{thehardestparttoboundlemma}
With expectations taken over $\pi$,
\begin{eqnarray}
\sum_{i=1}^{\lfloor|B(z,\delta nr)|/2\rfloor}\,
\mathop{\mathrm E}\left[\,
d^2\left(\pi\left(2i-1\right),\pi\left(2i\right)\right)
\,\right]
\le
\left(1+o(1)
\right)\left(\delta n^2rr'+2\delta^2 n r^2\right).
\label{boundforthedistancesquaressummed}
\end{eqnarray}
\end{lemma}
\begin{proof}
By line~5 of {\sf Las Vegas median},
$\{\pi(2i-1),\pi(2i)\}$ is a uniformly random size-$2$ subset of
$B(z,\delta nr)$ for each
$i\in[\,\lfloor|B(z,\delta nr)|/2\rfloor\,]$.
Therefore,
\begin{eqnarray}
&&\sum_{i=1}^{\lfloor|B(z,\delta nr)|/2\rfloor}\,
\mathop{\mathrm E}\left[\,
d^2\left(\pi\left(2i-1\right),\pi\left(2i\right)\right)
\,\right]\label{startofobjective}\\
&=&\sum_{i=1}^{\lfloor|B(z,\delta nr)|/2\rfloor}\,
\frac{1}{|B(z,\delta nr)|\cdot(|B(z,\delta nr)|-1)}
\cdot
\sum_{\text{distinct $u$, $v\in B(z,\delta nr)$}}\,d^2\left(u,v\right)\nonumber\\
&\le&\sum_{i=1}^{\lfloor|B(z,\delta nr)|/2\rfloor}\,
\frac{1}{|B(z,\delta nr)|\cdot(|B(z,\delta nr)|-1)}
\cdot
\sum_{u, v\in B(z,\delta nr)}\,d^2\left(u,v\right)
\nonumber\\
&=&\left\lfloor\frac{|B(z,\delta nr)|}{2}\right\rfloor\cdot
\frac{1}{|B(z,\delta nr)|\cdot(|B(z,\delta nr)|-1)}
\cdot
\sum_{u, v\in B(z,\delta nr)}\,d^2\left(u,v\right).
\nonumber
\end{eqnarray}
For all $u$, $v\in B(z,\delta nr)$,
\begin{eqnarray}
d\left(u,v\right)\le d\left(z,u\right)+d\left(z,v\right)
\le\delta nr+\delta nr=2\delta nr,
\label{constraint}
\end{eqnarray}
where the first inequality follows from the triangle inequality.

By equations~(\ref{smallerballaveragedefinition})~and~(\ref{startofobjective})--(\ref{constraint}),
the left-hand side of inequality~(\ref{boundforthedistancesquaressummed})
cannot exceed the optimal value of the following problem, called {\sc max square sum}:\\
\begin{quote}
Find
$d_{u,v}\in \mathbb{R}$
for all $u$, $v\in B(z,\delta nr)$ to maximize
\begin{eqnarray}
\left\lfloor\frac{|B(z,\delta nr)|}{2}\right\rfloor\cdot
\frac{1}{|B(z,\delta nr)|\cdot(|B(z,\delta nr)|-1)}
\cdot
\sum_{u, v\in B(z,\delta nr)}\,d_{u,v}^2
\label{objectiveofoptimization}
\end{eqnarray}
subject to
\begin{eqnarray}
\frac{1}{|B(z,\delta nr)|^2}\cdot
\sum_{u,v\in B(z,\delta nr)}\, d_{u,v}
=
r',\label{averagedistanceconstraint}\\
\forall
u, v\in B\left(z,\delta nr\right),\,\,
0\le
d_{u,v}
\le
2\delta nr.\label{largestdistanceconstraint}
\end{eqnarray}
\end{quote}
Above, constraint~(\ref{averagedistanceconstraint})
(resp., (\ref{largestdistanceconstraint}))
mimics equation~(\ref{smallerballaveragedefinition})
(resp., inequality~(\ref{constraint}) and the
non-negativeness of distances).
Appendix~\ref{analyzingthemaximizationproblem}
bounds
the
optimal value of
{\sc max square sum}
from
above by
\begin{eqnarray}
\left\lfloor\frac{|B(z,\delta nr)|}{2}\right\rfloor
\frac{1}{|B(z,\delta nr)|\cdot(|B(z,\delta nr)|-1)}
\cdot
\left(\left\lfloor\frac{|B(z,\delta nr)|^2 r'}{2\delta nr}\right\rfloor+1\right)
\cdot \left(2\delta nr\right)^2.
\nonumber
\end{eqnarray}
This evaluates to be
at most
$(1+o(1))(\delta n^2 rr'+2\delta^2 n r^2)$
by Lemma~\ref{thesmallradiusballislarge}.
\comment{ 
Finally,
$$\left(1\pm o(1)\right)\delta n^2 rr'\le 2\left(1+o(1)\right)\delta n^2 r^2$$
by Lemma~\ref{inneraverageandoverallaverage}.
}
\end{proof}

Recall that the variance of any random variable $X$
equals
$\mathop{\mathrm E}[X^2]-(\mathop{\mathrm E}[X])^2$.

\begin{lemma}\label{boundonthevarianceofthelengthofthematching}
With variances taken over $\pi$,
$$
\mathop{\mathrm{var}}\left(
\sum_{i=1}^{\lfloor|B(z,\delta nr)|/2\rfloor}\,
d\left(\pi\left(2i-1\right),\pi\left(2i\right)\right)
\right)
\le 2\left(1+o(1)
\right)\delta n^2 r^2.
$$
\end{lemma}
\begin{proof}
By equations~(\ref{startofthemeanofsquare})--(\ref{endofthemeanofsquare})
and
Lemmas~\ref{distinctdistancesproduct}--\ref{thehardestparttoboundlemma},
{\small 
\begin{eqnarray*}
\mathop{\mathrm E}\left[\,
\left(\sum_{i=1}^{\lfloor|B(z,\delta nr)|/2\rfloor}\,
d\left(\pi\left(2i-1\right),\pi\left(2i\right)\right)\right)^2
\,\right]
\le
\frac{1}{4}\cdot \left(1
+
o(1)
\right) n^2 \left(r'\right)^2
+\left(1+o(1)
\right)\left(\delta n^2 rr'+2\delta^2 nr^2\right).
\end{eqnarray*}
}
This and
Lemma~\ref{squareofthemeanlemma} imply
$$
\mathop{\mathrm{var}}\left(
\sum_{i=1}^{\lfloor|B(z,\delta nr)|/2\rfloor}\,
d\left(\pi\left(2i-1\right),\pi\left(2i\right)\right)
\right)
\le
o(1)
\cdot
n^2\left(r'\right)^2
+\left(1+o(1)
\right)\left(\delta n^2 rr'+2\delta^2 nr^2\right).
$$
Finally, invoke Lemma~\ref{inneraverageandoverallaverage}.
\end{proof}

\begin{lemma}\label{randommatchingconcentrationlemma}
For
all
$k>1$,
$$
\Pr\left[\,
\left|\,
\left(
\sum_{i=1}^{\lfloor|B(z,\delta nr)|/2\rfloor}\,
d\left(\pi\left(2i-1\right),\pi\left(2i\right)\right)
\right)
-
\frac{1}{2}\cdot\left(1\pm o(1)
\right)nr'
\,\right|
\ge
k\sqrt{2\left(1+o(1)
\right)\delta}\, nr
\,\right]
\le\frac{1}{k^2},
$$
where the probability is taken over $\pi$.
\end{lemma}
\begin{proof}
Use
Chebyshev's inequality
and
Lemmas~\ref{squareofthemeanlemma}~and~\ref{boundonthevarianceofthelengthofthematching}.
\end{proof}

Let $z'\in B(z,\delta nr)$ be a $1$-median of $B(z,\delta nr)$, i.e.,
\begin{eqnarray}
z'=\mathop{\mathrm{argmin}}_{y\in B(z,\delta nr)}\,
\sum_{x\in B(z,\delta nr)}\, d\left(y,x\right),
\nonumber
\end{eqnarray}
breaking ties arbitrarily.
So by the averaging argument,
\begin{eqnarray}
\sum_{x\in B(z,\delta nr)}\,d\left(z',x\right)
\le
\frac{1}{|B(z,\delta nr)|}
\cdot\sum_{y\in B(z,\delta nr)}\,\sum_{x\in B(z,\delta nr)}\,d\left(y,x\right).
\label{simpleaveraging}
\end{eqnarray}

\begin{lemma}\label{1medianinthesmallerballhasatmosttheaveragedistance}
$$\sum_{x\in B(z,\delta nr)}\,d\left(z',x\right)\le
nr'.$$
\end{lemma}
\begin{proof}
We have
\begin{eqnarray}
\sum_{x\in B(z,\delta nr)}\,d\left(z',x\right)
\stackrel{\text{(\ref{simpleaveraging})}}{\le}
\frac{1}{|B(z,\delta nr)|}\cdot \sum_{u,v\in B(z,\delta nr)}\,d\left(u,v\right)
\stackrel{\text{(\ref{smallerballaveragedefinition})}}{=}
\left|B\left(z,\delta nr\right)\right|\cdot r'.
\nonumber
\end{eqnarray}
Clearly, $|B\left(z,\delta nr\right)|\le n$.
\end{proof}

\begin{lemma}\label{theinner1medianisclosetotheoverall1median}
For all sufficiently large $n$,
$$d\left(z',z\right)\le8r.$$
\end{lemma}
\begin{proof}
We have
\begin{eqnarray}
\sum_{x\in B(z,\delta nr)}\,d\left(z',x\right)
&\ge&\sum_{x\in B(z,\delta nr)}\,\left(d\left(z',z\right)-d\left(z,x\right)\right)
\label{firstinequalityprovingthemedianofthesmallballtobeclosetothatofthelargeball}\\
&\ge&\sum_{x\in B(z,\delta nr)}\,d\left(z',z\right)
-\sum_{x\in [n]}\,d\left(z,x\right)\nonumber\\
&=&\left(\sum_{x\in B(z,\delta nr)}\,d\left(z',z\right)\right)
-nr\nonumber\\
&=&\left|\,B\left(z,\delta nr\right)\,\right|\cdot d\left(z',z\right)
-nr,\nonumber
\end{eqnarray}
where the first
inequality (resp., the first
equality) follows from
the triangle inequality (resp.,
line~4 of {\sf Las Vegas median}).
By
Lemmas~\ref{inneraverageandoverallaverage}~and~\ref{1medianinthesmallerballhasatmosttheaveragedistance},
\begin{eqnarray}
\sum_{x\in B(z,\delta nr)}\,d\left(z',x\right)
\le 2
nr.
\label{secondinequalityprovingthemedianofthesmallballtobeclosetothatofthelargeball}
\end{eqnarray}
By
inequalities~(\ref{firstinequalityprovingthemedianofthesmallballtobeclosetothatofthelargeball})--(\ref{secondinequalityprovingthemedianofthesmallballtobeclosetothatofthelargeball})
and
Lemma~\ref{thesmallradiusballislarge},
$d(z',z)\le (3+ o(1))r$.\footnote{In fact, this is stronger than the lemma to be proved.}
\comment{ 
If $z'\notin B(z,8r)$, then
\begin{eqnarray*}
\sum_{x\in B(z,\delta nr)}\,d\left(z',z\right)
\ge 8r\cdot\left|\,B\left(z,\delta nr\right)\,\right|
\stackrel{\text{Lemma~\ref{thesmallradiusballislarge}}}{=}
8\left(1-o(1)
\right)nr,
\end{eqnarray*}
contradicting
inequalities~(\ref{firstinequalityprovingthemedianofthesmallballtobeclosetothatofthelargeball})--(\ref{secondinequalityprovingthemedianofthesmallballtobeclosetothatofthelargeball})
for all sufficiently large $n$.
}
\end{proof}

\begin{lemma}\label{innerballmediantotaldistance}
For all sufficiently large $n$,
$$
\sum_{x\in [n]}\,d\left(z',x\right)
\le
nr'+\frac{16r}{\delta}
+\sum_{x\in [n]\setminus B(z,\delta nr)}\,\left(d\left(z,x\right)-8r\right).
$$
\end{lemma}
\begin{proof}
By the triangle inequality,
\begin{eqnarray}
\sum_{x\in[n]\setminus B(z,\delta nr)}\,d\left(z',x\right)
&\le& \sum_{x\in[n]\setminus B(z,\delta nr)}\,\left(d\left(z',z\right)+d\left(z,x\right)\right)
\nonumber\\
&\stackrel{\text{Lemma~\ref{theinner1medianisclosetotheoverall1median}}}{\le}&
\sum_{x\in[n]\setminus B(z,\delta nr)}\,\left(8r+d\left(z,x\right)\right)
\nonumber\\
&\stackrel{\text{Lemma~\ref{thesmallradiusballislarge}}}{\le}&
\frac{16r}{\delta}
+
\sum_{x\in [n]\setminus B(z,\delta nr)}\,\left(d\left(z,x\right)-8r\right).
\nonumber
\end{eqnarray}
Now sum up the above with the inequality in
Lemma~\ref{1medianinthesmallerballhasatmosttheaveragedistance}.
\end{proof}

\begin{lemma}\label{probabilityofsatisfyingthecondition}
For all sufficiently large $n$ and
with
probability
greater than
$1/2$,
\begin{eqnarray}
\sum_{i=1}^{\lfloor|B(z,\delta nr)|/2\rfloor}\,
d\left(\pi\left(2i-1\right),\pi\left(2i\right)\right)
+\sum_{x\in [n]\setminus B\left(z,\delta nr\right)}\,\left(d(z,x)-8r\right)
\ge\frac{1-100\sqrt{\delta}}{2}\cdot nr,
\label{theconditiontobetested}
\end{eqnarray}
where the probability is taken over $\pi$ and the internal coin tosses of {\sf Indyk median}
in line~3 of {\sf Las Vegas median}.
\end{lemma}
\begin{proof}
By Lemma~\ref{randommatchingconcentrationlemma} with
$k=5$,
\begin{eqnarray}
\sum_{i=1}^{\lfloor|B(z,\delta nr)|/2\rfloor}\,
d\left(\pi\left(2i-1\right),\pi\left(2i\right)\right)
>
\frac{1}{2}\cdot\left(1\pm o(1)\right)nr'
-5\sqrt{2\left(1+o(1)\right)\delta}\, nr
\label{matchingislarge}
\end{eqnarray}
with probability at least
$1-1/25$.
By Fact~\ref{Indykfact} and line~3 of {\sf Las Vegas median},
\begin{eqnarray}
\sum_{x\in[n]}\,d\left(z,x\right)
&\le&\left(1+\frac{\epsilon}{10^{10}}\right)\cdot\min_{y\in[n]}\,\sum_{x\in[n]}\,d\left(y,x\right)
\label{thenearoptimalisreal1}\\
&\le&
\left(1+\frac{\epsilon}{10^{10}}\right)\cdot\sum_{x\in[n]}\,d\left(z',x\right)
\label{thenearoptimalisreal}
\end{eqnarray}
with probability at least
$1-1/e$.
Now by the union bound,
inequalities~(\ref{matchingislarge})--(\ref{thenearoptimalisreal}) hold
simultaneously
with probability at least
$1-1/25-1/e>1/2$.
It remains to derive inequality~(\ref{theconditiontobetested})
from inequalities~(\ref{matchingislarge})--(\ref{thenearoptimalisreal})
for all sufficiently large $n$.

Line~4 of {\sf Las Vegas median},
inequalities~(\ref{thenearoptimalisreal1})--(\ref{thenearoptimalisreal}) and
Lemma~\ref{innerballmediantotaldistance}
give
\begin{eqnarray}
nr
\le
\left(1+\frac{\epsilon}{10^{10}}\right)\left(
nr'+\frac{16r}{\delta}
+\sum_{x\in [n]\setminus B(z,\delta nr)}\,\left(d\left(z,x\right)-8r\right)
\right).
\end{eqnarray}
This and inequality~(\ref{matchingislarge}) imply
{\small 
\begin{eqnarray}
&&nr\nonumber\\
&\le&
\left(1+\frac{\epsilon}{10^{10}}\right)\left(
%
2\left(1\pm o(1)\right)\left[
\left(\sum_{i=1}^{\lfloor|B(z,\delta nr)|/2\rfloor}\,
d\left(\pi(2i-1),\pi(2i)\right)\right)
+5\sqrt{2\left(1+o(1)\right)\delta}\, nr
\right]
\right.\nonumber\\
&&\left.
+\frac{16r}{\delta}
+\sum_{x\in [n]\setminus B(z,\delta nr)}\,\left(d\left(z,x\right)-8r\right)
\right).
\label{thefinalinequalitydirtyform}
\end{eqnarray}
}
\footnote{To see this, rewrite inequality~(\ref{matchingislarge})
as
$$nr'<
2\left(1\pm o(1)\right)
\left[
\left(\sum_{i=1}^{\lfloor|B(z,\delta nr)|/2\rfloor}\,
d\left(\pi\left(2i-1\right),\pi\left(2i\right)\right)\right)
+5\sqrt{2\left(1+o(1)\right)\delta}\, nr
\right].$$
}
Clearly,
$16r/\delta\le 0.01\cdot\sqrt{\delta}\,nr$ for all sufficiently large $n$.
So inequality~(\ref{thefinalinequalitydirtyform})
implies,
for all sufficiently large $n$ and after laborious calculations,
{\footnotesize 
\begin{eqnarray*}
&&nr-\left(1+\frac{\epsilon}{10^{10}}\right)
11
\sqrt{2\left(1+o(1)\right)\delta}\,nr\\
&\le& \left(2+\frac{2\epsilon}{10^{10}}\right)\left(1+o(1)\right)
\sum_{i=1}^{\lfloor|B(z,\delta nr)|/2\rfloor}\,d\left(\pi(2i-1),\pi(2i)\right)
+\left(1+\frac{\epsilon}{10^{10}}\right)
\cdot
\sum_{x\in[n]\setminus B(z,\delta nr)}\,\left(d\left(z,x\right)-8r\right).
\end{eqnarray*}
}
This implies inequality~(\ref{theconditiontobetested})
for all sufficiently large $n$
(note that
$\epsilon/10^{10}<\sqrt{\delta}$ by line~1
of {\sf Las Vegas Median}).\footnote{Divide both sides by $(2+2\epsilon/10^{10})(1+o(1))$ so that
the coefficient before $\sum_{i=1}^{\lfloor|B(z,\delta nr)|/2\rfloor}\,d(\pi(2i-1),\pi(2i))$
becomes $1$ in the
right-hand side. Then verify the left-hand side
(which is now $(nr-(1+\epsilon/10^{10})11\sqrt{2(1+o(1))\delta}\,nr)/((2+2\epsilon/10^{10})(1+o(1)))$)
to be
at least
$(1-100\sqrt{\delta})nr/2$
for all sufficiently large $n$.}
\end{proof}

Lemma~\ref{probabilityofsatisfyingthecondition} and lines~6--7 of {\sf Las Vegas median}
show the probability of termination in any iteration to be
$\Omega(1)$.
Because
the proof of Lemma~\ref{probabilityofsatisfyingthecondition}
implies
that
inequalities~(\ref{theconditiontobetested})--(\ref{thenearoptimalisreal})
hold simultaneously with probability $\Omega(1)$ in any iteration of {\sf Las Vegas median},
it happens with probability $\Omega(1)$ that
in the first iteration,
$z$ is returned in
line~7 (because of inequality~(\ref{theconditiontobetested}))
and
is $(1+\epsilon/10^{10})$-approximate (because of inequality~(\ref{thenearoptimalisreal1})).
So {\sf Las Vegas median}
outputs a $(1+\epsilon/10^{10})$-approximate
$1$-median
with probability
$\Omega(1)$ in the first iteration.
In summary, we have the following.

\begin{lemma}\label{MonteCarloaspectlemma}
The first iteration of the {\bf while} loop of
{\sf Las Vegas median} outputs
a $(1+\epsilon)$-approximate $1$-median
with probability
$\Omega(1)$.
\end{lemma}

\section{Putting things together}\label{mainresultsection}

We now show that
{\sc metric $1$-median} has a Las Vegas $(2+\epsilon)$-approximation algorithm
with an expected $O(n/\epsilon^2)$ running time for all constants $\epsilon>0$.
Our algorithm
also outputs
a $(1+\epsilon)$-approximate
$1$-median in time $O(n/\epsilon^2)$ with probability $\Omega(1)$.

\begin{theorem}\label{maintheorem}
For
each constant
$\epsilon>0$,
{\sc metric $1$-median} has a randomized algorithm that
(1)~{\em always} outputs a
$(2+\epsilon)$-approximate solution
in an expected
$O(n/\epsilon^2)$
time and that
(2)~outputs a $(1+\epsilon)$-approximate solution in time $O(n/\epsilon^2)$
with probability $\Omega(1)$.
\end{theorem}
\begin{proof}
By
Lemma~\ref{wehaveagoodsolutionattheend},
{\sf Las Vegas median}
outputs a
$(2+\epsilon)$-approximate
$1$-median
at termination.
To prevent {\sf Las Vegas median}
from running forever,
find a $1$-median by brute force (which obviously takes $O(n^2)$ time)
after
$n^2$ steps of computation.

By Fact~\ref{Indykfact}, line~3
of
{\sf Las Vegas median}
takes $O(n/\epsilon^2)$ time.
Line~5
takes
time $O(|B(z,\delta nr)|)=O(n)$
by
the Knuth shuffle.
Clearly, the other lines also take $O(n)$ time.
Consequently, each
iteration of the {\bf while} loop
of {\sf Las Vegas median}
takes $O(n/\epsilon^2)$ time.
By Lemma~\ref{probabilityofsatisfyingthecondition} and lines~6--7,
{\sf Las Vegas median}
runs for at most $1/\Omega(1)=O(1)$ iterations in expectation.
So
its
expected running time
is
$O(1)\cdot O(n/\epsilon^2)=O(n/\epsilon^2)$.

Having shown
each iteration
of {\sf Las Vegas median}
to take $O(n/\epsilon^2)$ time,
establish condition~(2) of the theorem with
Lemma~\ref{MonteCarloaspectlemma}.
\end{proof}

By
Fact~\ref{Indykfact},
{\sf Indyk median} satisfies condition~(2)
in Theorem~\ref{maintheorem}.
But it does not satisfy condition~(1).

We briefly
justify
the optimality of
the ratio of $2+\epsilon$ in
Theorem~\ref{maintheorem}.
Let
$A$ be a randomized algorithm
that
always
outputs
a
$(2-\epsilon)$-approximate
$1$-median.
Furthermore, denote by $p\in [n]$ (resp., $Q\subseteq [n]\times [n]$)
the output (resp., the set of queries as unordered pairs)
of $A^{d_1}(n)$, where $d_1$ is the discrete metric (i.e.,
$d_1(x,y)=1$ and $d_1(x,x)=0$ for all distinct $x$, $y\in [n]$).
Without loss of generality, assume $(p,y)\in Q$ for all $y\in [n]\setminus\{p\}$ by adding dummy
queries.
So
$A$ knows
that
\begin{eqnarray}
\sum_{y\in [n]}\,d_1\left(p,y\right)=n-1.
\label{outputunderthediscretemetric}
\end{eqnarray}
Furthermore, assume that $A$ never queries for the distance from a point to itself.

In the sequel,
consider the case
that $|Q|<\epsilon\cdot(n-1)^2/4$.
By
the averaging argument, there exists a point $\hat{p}\in [n]\setminus \{p\}$
involved in at most $2\cdot|Q|/(n-1)$ queries in $Q$.
Clearly, $A$ cannot exclude the possibility that $d_1(\hat{p},y)=1/2$ for all $y\in[n]\setminus\{\hat{p}\}$
satisfying $(\hat{p},y)\notin Q$.
In summary,
$A$ cannot rule out the case that
\begin{eqnarray}
\sum_{y\in[n]}\,d_1\left(\hat{p},y\right)&\le& \frac{2\cdot |Q|}{n-1}\cdot 1
+\left(n-1-\frac{2\cdot |Q|}{n-1}\right)\cdot \frac{1}{2}
< \left(\frac{1}{2}+\frac{\epsilon}{4}\right)\cdot(n-1).\,\,\,\,\,
\label{acasethatcannotberuledout}
\end{eqnarray}
Equations~(\ref{outputunderthediscretemetric})--(\ref{acasethatcannotberuledout})
contradict
the guarantee that $p$ is $(2-\epsilon)$-approximate.
In summary,
any randomized algorithm that always outputs a $(2-\epsilon)$-approximate
$1$-median must {\em always} make at least $\epsilon\cdot (n-1)^2/4=\Omega(\epsilon n^2)$ queries
given oracle access to the discrete metric.


\appendix
\section{Analyzing {\sc max square sum}}\label{analyzingthemaximizationproblem}

\setcounter{theorem}{0}
\numberwithin{theorem}{section}

{\sc Max square sum} has an optimal solution, denoted
$\{\tilde{d}_{u,v}\in\mathbb{R}\}_{u,v\in B(z,\delta nr)}$,
because
its
feasible solutions
(i.e., those satisfying
constraints~(\ref{averagedistanceconstraint})--(\ref{largestdistanceconstraint}))
form a
closed and bounded
subset of
$\mathbb{R}^{(|B(z,\delta nr)|^2)}$.
(Recall from elementary mathematical analysis that a continuous
real-valued function on a
closed and bounded
subset of $\mathbb{R}^k$
has a maximum value, where $k<\infty$.)
Note that
$\{\tilde{d}_{u,v}\in\mathbb{R}\}_{u,v\in B(z,\delta nr)}$
must be feasible to {\sc max square sum}.
Below is a consequence of constraint~(\ref{averagedistanceconstraint}).

\begin{lemma}\label{maximumnumberoflargestvaluevariables}
\begin{eqnarray}
\left|
\left\{
\left(u,v\right)\in
B^2\left(z,\delta nr\right)
\mid \tilde{d}_{u,v}=2\delta nr
\right\}
\right|
\le \left\lfloor\frac{|B(z,\delta nr)|^2r'}{2\delta nr}\right\rfloor.
\label{maximumnumberoflargestvaluevariablesinequality}
\end{eqnarray}
\end{lemma}
\begin{proof}
Clearly,
$$
\left|B(z,\delta nr)\right|^2 r'
\stackrel{\text{(\ref{averagedistanceconstraint})}}{=}
\sum_{u,v\in B(z,\delta nr)}\,\tilde{d}_{u,v}
\ge
\left|
\left\{
\left(u,v\right)\in
B^2\left(z,\delta nr\right)
\mid \tilde{d}_{u,v}=2\delta nr
\right\}
\right|
\cdot 2\delta nr.
$$
Furthermore, the left-hand side of
inequality~(\ref{maximumnumberoflargestvaluevariablesinequality})
is an integer.
\end{proof}

\begin{lemma}\label{supportofoptimalsolution}
$$
\left|
\left\{
\left(u,v\right)\in
B^2\left(z,\delta nr\right)
\mid \tilde{d}_{u,v}>0
\right\}
\right|
\le \left\lfloor\frac{|B(z,\delta nr)|^2r'}{2\delta nr}\right\rfloor+1.
$$
\end{lemma}
\begin{proof}
Assume otherwise.
Then
\begin{eqnarray*}
&&\left|
\left\{
\left(u,v\right)\in
B^2\left(z,\delta nr\right)
\mid \left(\tilde{d}_{u,v}>0\right)\land\left(\tilde{d}_{u,v}\neq 2\delta nr\right)
\right\}
\right|\\
&\ge&
\left|
\left\{
\left(u,v\right)\in
B^2\left(z,\delta nr\right)
\mid \tilde{d}_{u,v}>0
\right\}
\right|
-\left|
\left\{
\left(u,v\right)\in
B^2\left(z,\delta nr\right)
\mid \tilde{d}_{u,v}=2\delta nr
\right\}
\right|\\
&\ge&
\left\lfloor\frac{|B(z,\delta nr)|^2r'}{2\delta nr}\right\rfloor+2
-\left|
\left\{
\left(u,v\right)\in
B^2\left(z,\delta nr\right)
\mid \tilde{d}_{u,v}=2\delta nr
\right\}
\right|\\
&\stackrel{\text{Lemma~\ref{maximumnumberoflargestvaluevariables}}}{\ge}&
2.
\end{eqnarray*}
So by
constraint~(\ref{largestdistanceconstraint}) (and the feasibility of
$\{\tilde{d}_{u,v}\}_{u,v\in B(z,\delta nr)}$ to {\sc max square sum}),
$$
\left|
\left\{
\left(u,v\right)\in
B^2\left(z,\delta nr\right)
\mid 0<\tilde{d}_{u,v}<2\delta nr
\right\}
\right|
\ge 2.
$$
Consequently,
there exist distinct $(x,y)$,
$(x',y')\in B^2(z,\delta nr)$
satisfying
\begin{eqnarray}
0<\tilde{d}_{x,y},\, \tilde{d}_{x',y'}<2\delta nr.\label{thenonfullvariable1}
\end{eqnarray}
By symmetry, assume
$\tilde{d}_{x,y}\ge \tilde{d}_{x',y'}$.
By
inequality~(\ref{thenonfullvariable1}),
there exists a small real number $\beta>0$
such that
increasing $\tilde{d}_{x,y}$ by $\beta$ and simultaneously
decreasing $\tilde{d}_{x',y'}$ by $\beta$
will preserve
constraints~(\ref{averagedistanceconstraint})--(\ref{largestdistanceconstraint}).
I.e., the solution $\{\hat{d}_{u,v}\in\mathbb{R}\}_{u,v\in B(z,\delta nr)}$ defined below is
feasible
to
{\sc max square sum}:
\begin{eqnarray}
\hat{d}_{u,v}=
\left\{
\begin{array}{ll}
\tilde{d}_{x,y}+\beta, & \text{if $(u,v)=(x,y)$},\\
\tilde{d}_{x',y'}-\beta, & \text{if $(u,v)=(x',y')$},\\
\tilde{d}_{u,v},&\text{otherwise}.
\end{array}
\right.\label{variatedsolution}
\end{eqnarray}

Clearly,
objective~(\ref{objectiveofoptimization})
w.r.t.\ $\{\hat{d}_{u,v}\}_{u,v\in B(z,\delta nr)}$
exceeds that w.r.t.\
$\{\tilde{d}_{u,v}\}_{u,v\in B(z,\delta nr)}$
by
\begin{eqnarray*}
&&\left\lfloor\frac{|B(z,\delta nr)|}{2}\right\rfloor\cdot
\frac{1}{|B(z,\delta nr)|\cdot(|B(z,\delta nr)|-1)}
\cdot
\sum_{u,v\in B(z,\delta nr)}\,\left({\hat{d}}^2_{u,v}-{\tilde{d}}^2_{u,v}
\right)\\
&\stackrel{\text{(\ref{variatedsolution})}}{=}&
\left\lfloor\frac{|B(z,\delta nr)|}{2}\right\rfloor\cdot
\frac{1}{|B(z,\delta nr)|\cdot(|B(z,\delta nr)|-1)}\\
&\cdot&
\left(
\left(\tilde{d}_{x,y}+\beta\right)^2+
\left(\tilde{d}_{x',y'}-\beta\right)^2
-{\tilde{d}}^2_{x,y}-{\tilde{d}}^2_{x',y'}
\right)\\
&=&
\left\lfloor\frac{|B(z,\delta nr)|}{2}\right\rfloor\cdot
\frac{1}{|B(z,\delta nr)|\cdot(|B(z,\delta nr)|-1)}
\cdot
\left(
2\beta\tilde{d}_{x,y}-2\beta\tilde{d}_{x',y'}
+2\beta^2
\right)\\
&>&0,
\end{eqnarray*}
where the inequality holds
because $\tilde{d}_{x,y}\ge \tilde{d}_{x',y'}$ and
$\beta>0$.

In summary,
$\{\hat{d}_{u,v}\}_{u,v\in B(z,\delta nr)}$ is a feasible solution
achieving a greater
objective~(\ref{objectiveofoptimization})
than the optimal solution
$\{\tilde{d}_{u,v}\}_{u,v\in B(z,\delta nr)}$ does, a contradiction.
\end{proof}

We now
bound the optimal value of
{\sc max square sum}.

\begin{theorem}
The optimal value of {\sc max square sum}
is at most
$$
\left\lfloor\frac{|B(z,\delta nr)|}{2}\right\rfloor\cdot
\frac{1}{|B(z,\delta nr)|\cdot(|B(z,\delta nr)|-1)}
\cdot
\left(
\left\lfloor
\frac{|B(z,\delta nr)|^2r'}{2\delta nr}
\right\rfloor+1
\right)
\cdot \left(2\delta nr\right)^2
$$
\end{theorem}
\begin{proof}
W.r.t.\
the optimal (and thus feasible) solution $\{\tilde{d}_{u,v}\}_{u,v\in B(z,\epsilon nr)}$,
objective~(\ref{objectiveofoptimization}) equals
\begin{eqnarray*}
&&\left\lfloor\frac{|B(z,\delta nr)|}{2}\right\rfloor\cdot
\frac{1}{|B(z,\delta nr)|\cdot(|B(z,\delta nr)|-1)}
\cdot
\sum_{u,v\in B(z,\delta nr)}\,
\chi\left[\tilde{d}_{u,v}\neq 0\right]\cdot {\tilde{d}}^2_{u,v}\\
&\stackrel{\text{(\ref{largestdistanceconstraint})}}{\le}&
\left\lfloor\frac{|B(z,\delta nr)|}{2}\right\rfloor\cdot
\frac{1}{|B(z,\delta nr)|\cdot(|B(z,\delta nr)|-1)}
\cdot
\sum_{u,v\in B(z,\delta nr)}\,
\chi\left[\tilde{d}_{u,v}>0\right]\cdot \left(2\delta nr\right)^2,
\end{eqnarray*}
where $\chi[P]=1$ if $P$ is true and $\chi[P]=0$ otherwise, for any
predicate $P$.
Now invoke Lemma~\ref{supportofoptimalsolution}.
\end{proof}

\comment{ 
By the triangle inequality,
\begin{eqnarray*}
&&\sum_{\text{distinct $u$, $v\in B(z,\delta nr)$}}\,
d\left(u,v\right)\\
&\le&
\sum_{\text{distinct $u$, $v\in B(z,\delta nr)$}}\,
\left(d\left(z,u\right)+d\left(z,v\right)\right)\\
&=&
\left(\left|B(z,\delta nr)\right|-1\right)\cdot
\sum_{u\in B(z,\delta nr)}\,
d\left(z,u\right)
+\left(\left|B(z,\delta nr)\right|-1\right)\cdot
\sum_{v\in B(z,\delta nr)}\,
d\left(z,v\right).
\end{eqnarray*}
}

\bibliographystyle{plain}
\bibliography{las_vegas_median}

\begin{thebibliography}{10}

\bibitem{Cha13}
C.-L. Chang.
\newblock Deterministic sublinear-time approximations for metric $1$-median
  selection.
\newblock {\em Information Processing Letters}, 113(8):288--292, 2013.

\bibitem{Cha15}
C.-L. Chang.
\newblock A deterministic sublinear-time nonadaptive algorithm for metric
  $1$-median selection.
\newblock {\em Theoretical Computer Science}, 602:149--157, 2015.

\bibitem{Cha16COCOON}
C.-L. Chang.
\newblock Metric $1$-median selection: Query complexity vs.\ approximation
  ratio.
\newblock In {\em Proceedings of the 22nd International Computing and
  Combinatorics Conference}, pages 131--142, Ho Chi Minh City, Vietnam, 2016.
\newblock Full version at \url{https://arxiv.org/abs/1509.05662}.

\bibitem{Cha17}
C.-L. Chang.
\newblock A lower bound for metric $1$-median selection.
\newblock {\em Journal of Computer and System Sciences}, 84:44--51, 2017.

\bibitem{EW04}
D.~Eppstein and J.~Wang.
\newblock Fast approximation of centrality.
\newblock {\em Journal of Graph Algorithms and Applications}, 8(1):39--45,
  2004.

\bibitem{GR08}
O.~Goldreich and D.~Ron.
\newblock Approximating average parameters of graphs.
\newblock {\em Random Structures \& Algorithms}, 32(4):473--493, 2008.

\bibitem{GMMMO03}
S.~Guha, A.~Meyerson, N.~Mishra, R.~Motwani, and L.~O'Callaghan.
\newblock Clustering data streams: \uppercase{T}heory and practice.
\newblock {\em IEEE Transactions on Knowledge and Data Engineering},
  15(3):515--528, 2003.

\bibitem{Ind99}
P.~Indyk.
\newblock Sublinear time algorithms for metric space problems.
\newblock In {\em Proceedings of the 31st Annual ACM Symposium on Theory of
  Computing}, pages 428--434, 1999.

\bibitem{Ind00}
P.~Indyk.
\newblock {\em High-dimensional computational geometry}.
\newblock PhD thesis, Stanford University, 2000.

\bibitem{KSS10}
A.~Kumar, Y.~Sabharwal, and S.~Sen.
\newblock Linear-time approximation schemes for clustering problems in any
  dimensions.
\newblock {\em Journal of the ACM}, 57(2):5, 2010.

\bibitem{MP04}
R.~R. Mettu and C.~G. Plaxton.
\newblock Optimal time bounds for approximate clustering.
\newblock {\em Machine Learning}, 56(1--3):35--60, 2004.

\bibitem{MR95}
R.~Motwani and P.~Raghavan.
\newblock {\em Randomized Algorithms}.
\newblock Cambridge University Press, Cambridge, UK, 1995.

\bibitem{Rud76}
W.~Rudin.
\newblock {\em Principles of Mathematical Analysis}.
\newblock McGraw-Hill, 3rd edition, 1976.

\bibitem{WF94}
S.~Wasserman and K.~Faust.
\newblock {\em Social Network Analysis: Methods and Applications}.
\newblock Cambridge University Press, 1994.

\bibitem{Wu14}
B.-Y. Wu.
\newblock On approximating metric $1$-median in sublinear time.
\newblock {\em Information Processing Letters}, 114(4):163--166, 2014.

\end{thebibliography}

\noindent

\end{document}